\documentclass[
 10pt,
 superscriptaddress,
 nofootinbib,
 amsmath,amssymb,
 aps,
 pra,
 notitlepage,
twocolumn
]{revtex4-1}

\usepackage{amsthm}
\usepackage{hyperref}
\usepackage[dvipsnames]{xcolor}
\usepackage{graphicx}% Include figure files
\usepackage{bm}% bold math
\usepackage{bbm}% math symbols
\usepackage{braket}

\usepackage{qcircuit}

\usepackage{listings}

\usepackage{hyperref}% add hypertext capabilities
\hypersetup{backref=true,       
    pagebackref=true,               
    hyperindex=true,                
    colorlinks=true,                
    breaklinks=true,                
    urlcolor= black,                
    linkcolor= blue,                
    bookmarks=true,                 
    bookmarksopen=false,
    filecolor=black,
    citecolor=blue,
    linkbordercolor=blue,
}

\newtheorem{thm}{Theorem}

\newtheorem{lem}{Lemma}
\newtheorem{cor}{Corollary}

\begin{document}

\title{Tailgating quantum circuits for high-order energy derivatives}

\author{Jack Ceroni}
\email{jack.ceroni@mail.utoronto.ca}
\affiliation{Xanadu, Toronto, ON, M5G 2C8, Canada}
\affiliation{Department of Mathematics, University of Toronto, Toronto, ON, M5S 3E1, Canada}
\author{Alain Delgado}
\affiliation{Xanadu, Toronto, ON, M5G 2C8, Canada}
\author{Soran Jahangiri}
\affiliation{Xanadu, Toronto, ON, M5G 2C8, Canada}
\author{Juan Miguel Arrazola}
\affiliation{Xanadu, Toronto, ON, M5G 2C8, Canada}

\date{\today}

\begin{abstract}
To understand the chemical properties of molecules, it is often important to study derivatives of energies with respect to nuclear coordinates or external fields. Quantum algorithms for computing energy derivatives have been proposed, but only limited work has been done to address the specific challenges that arise in this context, where calculations are more complicated and involve more stringent requirements on accuracy compared to single-point energy calculations. In this work, we introduce a technique to improve the performance of variational quantum circuits calculating energy derivatives. The method, which we refer to as tailgating, is an adaptive procedure that selects gates based on their gradient with respect to the expectation value of Hamiltonian derivatives. These gates are then added at the end of a quantum circuit originally designed to calculate ground- or excited-state energies. A distinguishing feature of this approach is that the appended gates do not need to be optimized: their parameters can be set to zero and varied only for the purpose of computing energy derivatives, via calculating derivatives with respect to circuit parameters. We support the validity of this method by establishing sufficient conditions for a circuit to compute accurate energy gradients. This is achieved through a connection between energy derivatives and eigenstates of Taylor approximations of the Hamiltonian. We illustrate the advantages of the tailgating approach by performing simulations calculating the vibrational modes of beryllium hydride and water: quantities that depend on second-order energy derivatives.

\end{abstract}

\maketitle

\section{Introduction}

Quantum algorithms have been studied as a potential avenue to address some of the challenges of accurately simulating the electronic structure of molecules~\cite{mcardle2020quantum, aspuru2005simulated, google2020hartree, lanyon2010towards}. Advanced quantum algorithms for performing these calculations require a large number of qubits and deep circuits~\cite{cao2019quantum, o2021efficient, su2021fault, low2019hamiltonian}, which makes them challenging to implement on currently available quantum hardware or classical simulators. Variational quantum algorithms (VQAs) have been instead considered as a platform to perform proof-of-principle simulations of molecules on existing devices. VQAs are a paradigm for designing quantum algorithms in which the structure of a quantum circuit is fixed and the gate parameters are optimized to minimize a suitable cost function, typically the expectation value of an observable with respect to the prepared wavefunction \cite{cerezo2021variational}.  

In the context of quantum chemistry, one of the most-studied algorithms is the variational quantum eigensolver (VQE)~\cite{kandala2017hardware, peruzzo2014variational}. Here, a quantum circuit is optimized in order to minimize the expectation value of a molecular Hamiltonian with respect to the output state of the circuit, yielding an approximation of the ground-state energy. Similar algorithms have been proposed to calculate other electronic properties such as excited-state energies~\cite{higgott2019variational, nakanishi2019subspace, mcclean2017hybrid, parrish2019quantum}, equilibrium geometries~\cite{delgado2021variational}, dipole transition amplitudes~\cite{ibe2020calculating}, transition states~\cite{azad2022quantum}, and molecular dynamics \cite{sokolov2021microcanonical}. It has also been shown that computing first- and second-order energy derivatives of ground and excited state energies can be performed using variational methods~\cite{mitarai2020theory,o2021efficient,azad2022quantum}. Second-order energy derivatives can be used to characterize the vibrational structure of molecules, which allows for the prediction of chemical properties such as vibronic spectra~\cite{huh2015boson}, simulating local mode dynamics~\cite{sparrow2018simulating, jahangiri2020quantum} and electron transport \cite{jahangiri2021quantum}, and computing thermodynamic observables~\cite{stober2022considerations}.

Beyond their scope of application, the techniques used for constructing variational quantum circuits have also become more sophisticated. One popular approach is the ADAPT-VQE method~\cite{grimsley2019adaptive, tang2021qubit, bilkis2021semi}. Starting from a pool of gates, this algorithm iteratively adds gates to a circuit by selecting those with large gradients with respect to the target cost function. This procedure has the advantage of being tailored to an individual Hamiltonian and typically leads to shorter-depth circuits without sacrificing accuracy. 
ADAPT-VQE has shown good results when used to compute energy eigenvalues of some fixed Hamiltonian~\cite{claudino2020benchmarking}, and circuits constructed with adaptive methods have been effective for molecular geometry optimization procedures in small molecules, where it is required to calculate first-order energy derivatives~\cite{delgado2021variational}. However, no prior work has been done to analyze the capability of adaptive circuits for computing second-order energy derivatives beyond the simple case of the hydrogen molecule~\cite{mitarai2020theory}.

In this work, we argue that standard adaptive procedures fail to produce circuits capable of calculating accurate high-order energy derivatives because they only take into account limited information about the problem Hamiltonian. We then address this issue by introducing the method of \emph{tailgating}: a procedure in which parameterized gates are adaptively added to the end of a variational circuit, with their ``optimized" variational parameters set to zero. The gates are selected by evaluating their gradient with respect to the expectation value of \textit{derivatives} of the Hamiltonian. They do not contribute to the preparation of the ground state, but do have an effect when calculating energy derivatives, increasing the accuracy of the calculations. This procedure is particularly suitable for adaptive circuits, but in principle any circuit can be tailgated to increase accuracy when calculating energy derivatives. 

We support the validity of this method by establishing sufficient conditions for a quantum circuit to calculate accurate energy derivatives. This is achieved through a connection between Taylor approximations of the Hamiltonian and energy derivatives. Finally, we illustrate the advantages of tailgating with numerical examples, calculating the energy Hessians of beryllium hydride and water molecules. We subsequently use the Hessians to compute the vibrational frequencies of each molecule. We demonstrate that, in both cases, tailgated circuits provide values that are considerably more accurate than those obtained from standard adaptive circuits.

\section{Computing Accurate Energy Derivatives}
\label{sec:background}

We begin by introducing notation that will be used throughout this work. Let $U(\theta)$ be a parameterized quantum circuit that prepares an output state  $|\psi(\theta)\rangle = U(\theta) |0\rangle$, where $\theta=(\theta_1, \theta_2, \ldots, \theta_M)$ are the parameters of the circuit. We use $H(R)$ to represent a parameterized Hamiltonian, where $R=(R_1, R_2, \ldots, R_N)$ is a vector of parameters. For concreteness, we focus on the case where $R$ represents the nuclear coordinates of a molecule. Let $\theta^*(R)$ denote the optimal parameters such that the output state $|\psi(\theta^*(R))\rangle$ closest approximates the ground state $\ket{\psi_0(R)}$. We also describe the corresponding ground-state energy as 
\begin{equation}
    E(R) = \langle \psi_0(R) | H(R) |\psi_0(R)\rangle.
\end{equation}
We focus on the ground state for simplicity, but the analysis described throughout this work applies equally to any eigenstate of the Hamiltonian. Finally, we define 

\begin{equation}
\tilde{E}(\theta, R) = \langle \psi(\theta) | H(R) | \psi(\theta) \rangle,
\end{equation}
and $\tilde{E}(R) = \tilde{E}(\theta^{*}(R), R)$ as the approximation of the energy eigenvalue resulting from the output state of the circuit. 
For simplicity, we use $|\psi(R)\rangle$ as a shorthand for $|\psi(\theta^{*}(R))\rangle$. We work under the assumption that the family of Hamiltonians $H(R)$ is non-degenerate, so we can truly speak of a single ground state.

After identifying and optimizing a circuit $U(\theta)$ such that $\tilde{E}(R_0)\approx E(R_0)$ for a given $R_0$, e.g., the equilibrium geometry of the molecule, our goal is to compute accurate $n$-th order energy derivatives at $R_0$. In other words, we want to ensure that

\begin{equation}
    \frac{\partial^{n} \tilde{E}(R_0)}{\partial R_{j_1} \cdots R_{j_n}} \approx \frac{\partial^{n} E(R_0)}{\partial R_{j_1} \cdots R_{j_n}},
\end{equation}
for some choice of $n$ and coordinates $\{R_{j_1}, \ \ldots, \ R_{j_n}\}$.

Now that we have introduced the necessary notation, we briefly discuss the computation of energy derivatives, and the associated difficulties of performing these calculations. Computing quantities which rely on $n$-th order energy derivatives, with variational quantum circuits, already suffer from the issue of high sampling cost. This occurs even in the case of second-order derivatives, where the number of samples required to compute a single second-order energy derivative to error $\epsilon$ scales as \cite{mitarai2020theory}

\begin{equation}
   N_{\text{samples}} = \tilde{O} \left( \frac{n^4 M^2}{\epsilon'^2} \right). 
\end{equation}
where $n$ is the number of qubits, $M$ is the number of gate parameters in the variational circuit, and $\epsilon'$ is an error parameter defined in Ref.~\cite{mitarai2020theory}. It depends linearly on $\epsilon$, and implicitly on $n$ and $M$. In this paper, we do not focus on the important problem of sample complexity, but rather on a different roadblock related to the choice of circuit ansatz. When computing energy derivatives, there will generally be contributions due to derivatives of the state with respect to the nuclear coordinates. This can be seen for example in the expression for second-order derivatives~\cite{azad2022quantum}:

\begin{align}
    \label{eqn:s}
    \frac{\partial^2 E(R)}{\partial R_i\partial R_j} &= \langle \psi_0(R) | \frac{\partial^2 H(R)}{\partial R_i\partial R_j} | \psi_0(R)\rangle \nonumber\\
    &+ 2\text{Re}\left[\langle \psi_0(R) | \frac{\partial H(R)}{\partial R_i} \frac{\partial |\psi_0(R)\rangle}{\partial R_j}\right],
\end{align}
which contains a contribution from the state derivative $\partial |\psi_0(R)\rangle/\partial R_j$. Expressions for derivatives of this form are derived in Appendix~\hyperref[sec:appb]{B}. This suggests, at a high level, that in order to evaluate energy derivatives, we need a circuit that can prepare approximate ground states not only at $R_0$, but also in some neighbourhood around $R_0$.

Guaranteeing that a particular circuit possesses this property is not straightforward. For standard adaptive methods that only select gates based on the Hamiltonian $H(R_0)$, we argue that there may be situations in which certain omitted gates have negligible gradient at $R_0$, but are required to produce accurate energy derivatives. More precisely, these omitted gates, while having vanishingly small gradients during standard adaptive algorithms like ADAPT-VQE, have considerable effects on the \textit{degree} of accuracy to which the state $|\psi(R)\rangle$ matches the state $|\psi_0(R)\rangle$, in some neighbourhood around $R_0$. This change in accuracy can affect state derivatives. For instance, suppose that

\begin{equation}
    |\psi(R)\rangle = U(\theta^{*}(R))|0\rangle = |\psi_0(R)\rangle + \epsilon(R) |\zeta(R) \rangle,
    \label{eq:error}
\end{equation}
\noindent
for all $R$ in a neighbourhood around $R_0$, and $\epsilon(R) = || |\psi(R)\rangle - |\psi_0(R)\rangle ||$. If $\epsilon(R)$ is sufficiently small for a neighbourhood around $R_0$, the circuit $U$ prepares a state $|\psi(R)\rangle$ which is a good approximation for $|\psi_0(R)\rangle$. However, it could be the case that the magnitude of $\partial \epsilon(R) |\zeta(R)\rangle /\partial R_j$ is large, implying that first derivatives of $|\psi(R)\rangle$ and $|\psi_0(R)\rangle$ differ significantly. Intuitively, even though $|\psi(R)\rangle$ is close to $|\psi_0(R)\rangle$, the error between the two states could change rapidly as a function of $R$. This could then lead to large differences in second-order energy derivatives, as is indicated by Eq.~\eqref{eqn:s}.

Evidence for this phenomenon occurring in circuits prepared by adaptive methods is illustrated in Fig.~\hyperref[fig:gr]{1}. Here, we plot the fidelity $|\langle \psi(R_0 + \delta e_H) | \psi_0(R_0 + \delta e_H) \rangle|^2$, where $|\psi_0(R)\rangle$ is the ground state of the BeH$_2$ molecule at $R$ and $\delta e_H$ is a vector which stretches the two hydrogen atoms a distance $\delta$ in opposite directions. We adaptively prepare a circuit $U$ with respect to the Hamiltonian $H(R_0)$ using a pool of single and double excitation gates \cite{arrazola2021universal}. We then perform VQE (to a pre-specified convergence criterion of gradients having magnitude less than $10^{-5}$) for a collection of $\delta$ ranging from $-0.75 \ \text{Bohr}$ to $0.75 \ \text{Bohr}$. The fidelity is close to $1$ for all $\delta$, but the derivative of the fidelity with respect to $\delta$ is non-zero and decreasing for large $\delta$.

\begin{figure}
\includegraphics[width=240pt]{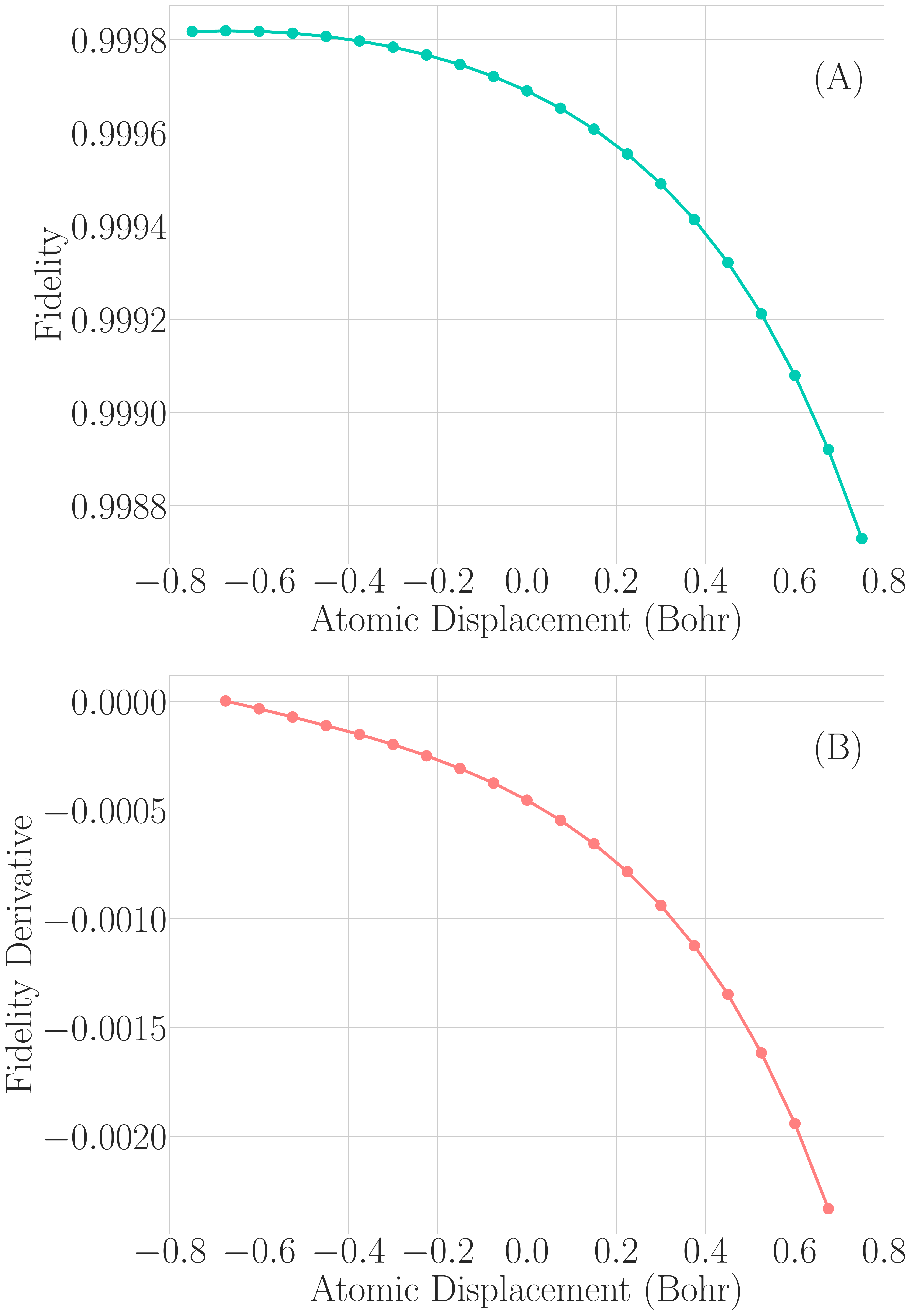}
\caption{(A) Fidelity $|\langle \psi(R_0 + \delta e_H) | \psi_0(R_0 + \delta e_H) \rangle|^2$ as a function of $\delta$: the distance the hydrogen atoms are displaced from equilibrium. Here we consider the fidelity between the ground state of the BeH$_2$ molecule and an approximation prepared by an optimized circuit yielded from ADAPT-VQE. (B) The derivative of the fidelity, computed using a central finite-difference.}
\label{fig:gr}
\end{figure}

Now, we briefly develop some intuition as to how this issue can be resolved by considering gradients of gates with respect to derivatives of the Hamiltonian $H(R)$. Consider an adaptive procedure, where gates $\{G_a\}$ are selected to build an ansatz $U$ which yields an approximation to the ground state of $H(R_0)$, based on the value of gradients of the form
\begin{equation}
    \label{eq:der}
    \frac{\partial}{\partial \theta_a}\bra{\psi} G_a(\theta_a)^{\dagger} H(R_0) G_a(\theta_a) \ket{\psi},
\end{equation}
where $|\psi\rangle$ is a state constructed during the previous iteration of the adaptive procedure. If the value of this gradient is large, then the gate $G_a$ is selected for the circuit ansatz which is used to prepare ground state of $H(R_0)$. Eventually, after iterating through all gates in $\{G_a\}$, we construct the ansatz $U$ which yields $|\psi(R_0)\rangle \approx |\psi_0(R_0)\rangle$. 

Now, for each $G_a$, consider the function

\begin{equation}
    \Delta_a(R) = \frac{\partial}{\partial \theta_a} \langle \psi(R_0) | G_a(\theta_a)^{\dagger} H(R) G_a(\theta_a) | \psi(R_0)\rangle |_{\theta_a = 0}.
\end{equation}
As in the original adaptive procedure, $\Delta_a(R)$ can be thought of as measuring the effect of the gate $G_a$ on an optimization procedure which attempts to minimize the expectation value of $H(R)$, starting from the state $|\psi(R_0)\rangle$. In general, we will expect $\Delta_a(R)$ to be small at $R = R_0$ and likely also in a neighbourhood around $R$, as adding any gate that was already selected or not previously selected will likely not have a large effect on improving the accuracy of the approximate ground state $|\psi(R_0)\rangle$. However, these gates could still provide small corrections to the state. As was indicated above, this is precisely what we should consider to ensure that first-order state derivatives are accurate. Thus, we turn our attention to the \textit{derivatives} of $\Delta_a(R)$ evaluated at $R = R_0$,
\vspace{-7pt}

\begin{equation}
    \frac{\partial \Delta_a(R_0)}{\partial R_j} = \frac{\partial}{\partial \theta_a} \langle \psi(R_0) | G_a(\theta_a)^{\dagger} \frac{\partial H(R_0)}{\partial R_j} G_a(\theta_a) | \psi(R_0)\rangle.
\end{equation}

If a derivative $\partial \Delta_a(R_0) / \partial R_j$ is large, this indicates that a gate $G_a$ contributes to the preparation of $|\psi(R_0)\rangle$ and $|\psi(R)\rangle$ for $R$ close to $R_0$ with significantly different magnitudes, relative to the distance from $R_0$ to $R$. If we add all gates of this form to the ansatz, our hope is that this will lead to more accurate state derivatives, which will in turn lead to more accurate energy derivatives. 

This treatment only takes into account first-order differences, so a natural subsequent question to ask is whether we have to consider higher-order derivatives of $\Delta_a(R)$ to prepare an ansatz which yields accurate higher-order energy derivatives. To answer this question in the affirmative, in the following sections we take a more rigorous approach, via Taylor series expansions of the Hamiltonian $H(R)$. 

\subsection{Taylor Series Expansions of Hamiltonians}

Keeping the importance of Hamiltonian derivatives in mind, we turn our attention to finding sufficient conditions for a circuit $U(\theta)$ that can approximate ground states $|\psi_0(R)\rangle$ in a neighbourhood around $R_0$ by Taylor-expanding $H(R)$ around $R_0$. In this form, we approximate $H(R)$ in terms of its derivatives at $R_0$. Let $H_{n}(R)$ be the $n$-th order Taylor approximation of $H$ centred at $R_0$:

\begin{align}
\label{eqn:taylor}
   & H_{n}(R) = H(R_0) + \displaystyle\sum_{j = 1}^N \frac{\partial H(R_0)}{\partial R_j} (R_j - R_{0j}) + \cdots \nonumber\\
   &+ \frac{1}{n!} \displaystyle\sum_{j_1, \ldots, j_n = 1}^N \frac{\partial^{n} H(R_0)}{\partial R_{j_1} \cdots \partial R_{j_{n}}} (R_{j_1} - R_{0j_1}) \cdots (R_{j_{n}} - R_{0j_n}). 
\end{align}

The Taylor expansion $H_n(R)$ provides a good approximation of $H(R)$ when $R$ is close to $R_0$, provided we truncate at sufficiently high order. Thus, it is reasonable to suspect that for sufficiently large $n$, the eigenvalues of $H_n$ closely approximate those of $H$ for a small neighbourhood around $R_0$. As a result, the energy derivatives computed with respect to the two Hamiltonians are equal to some order. This is formalized in the following theorem:

\begin{thm}
\label{thm:a}
 Let $|\phi_0(R)\rangle$ be the ground state of $H_{n - 1}(R)$. Suppose that there exists $\delta > 0$ such that a circuit $U(\theta)$ can prepare the state $|\phi_0(R)\rangle$ for all $R$ such that $0 < \| R - R_0 \| < \delta$. Then it holds that

\begin{equation}
\frac{\partial^{p} E(R_0)}{\partial R_{j_1} \cdots \partial R_{j_p}} = \frac{\partial^{p} \tilde{E}(R_0)}{\partial R_{j_1} \cdots \partial R_{j_p}},
\end{equation}
for all $p$ such that $0 \leq p \leq n$, and all choices of coordinates, $\{R_{j_1}, \ ..., \ R_{j_p}\}$.
\end{thm}
\noindent The proof is given in Appendix~\hyperref[sec:app]{A}. 

The implication of this result is crucial for our purposes: circuits that can compute accurate energy eigenstates of Taylor approximations of the Hamiltonian, in some neighbourhood around $R_0$, can be used to calculate accurate energy derivatives. Importantly, if the circuit can prepare accurate eigenstates for a Taylor expansion of order $n-1$, then it is possible to calculate accurate energy derivatives up to order $n$. In the case of $n = 1$, our result implies that to compute accurate first-order derivatives, it suffices that the circuit can prepare accurate ground states of the Hamiltonian $H(R_0)$ at the equilibrium geometry. 

This supports the observation that adaptive circuits are effective for computing first-order derivatives~\cite{delgado2021variational}. However, the sufficient condition of Theorem~\ref{thm:a} indicates that the ability to prepare ground states at equilibrium may not be enough for higher-order derivatives. In such cases, it is helpful to design circuits that can prepare ground states of higher-order Taylor approximations in the neighbourhood around equilibrium. This in turn can be achieved by focusing on the derivatives of the Hamiltonian, which determine the coefficients in the Taylor expansion. This is the main insight behind the tailgating algorithm that we describe next. 

\section{Tailgating}

%Making use of of Theorem~\hyperref[thm:a]{1}, the task of identifying a circuit that can compute accurate $n$-th order energy derivatives can be reduced to finding a circuit that can prepare $|\phi_k(R)\rangle$ for small $R$. The motivation for making this simplification is due to the fact that $H_{n - 1}(R)$ depends explicitly on the parameters $R$. As is apparent in Eq.~\eqref{eqn:taylor}, this Hamiltonian can be thought of as a linear combination of Hamiltonian derivatives, with coefficients given by $R$. On the other hand, $H(R)$ depends implicitly on the parameters $R$. This makes it more difficult to make claims about the behaviour of $H$ is a neighbourhood around $R_0$; difficulty that does not arise when considering $H_{n - 1}$.

\begin{figure}
\includegraphics[width=250pt]{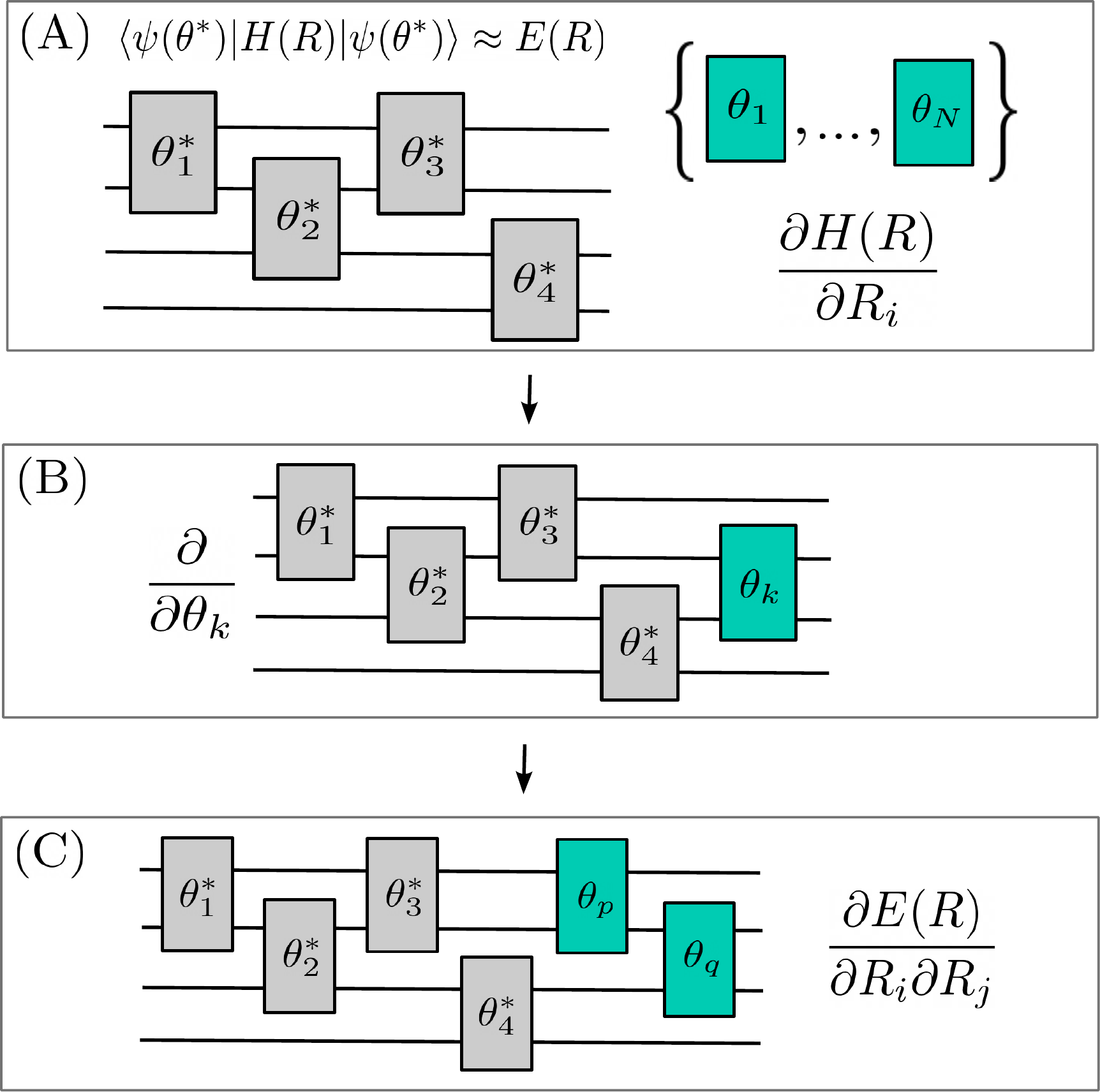}
\caption{An illustration of the tailgating procedure for the case of second-order energy derivatives. (A) We begin with a circuit which can prepare good approximations of the ground state energy of $H(R)$, along with a pool of gates and Hamiltonian derivatives. (B) We perform an adaptive selection of gates by iteratively applying each gate in the pool to the end of the circuit, and computing derivatives of this particular gate with respect to Hamiltonian derivatives. We add all gates with non-negligible gradients to the new, tailgated circuit. (C) This circuit can then be used to compute accurate second-order energy derivatives. Note that to calculate the desired energy derivatives, we evaluate derivatives of the added gates at $\theta_p = \theta_q = 0$.}
\label{fig:tail}
\end{figure}

Similar to the expression in Eq.~\eqref{eq:der}, consider an adaptive procedure where gates $\{G_a\}$ are selected based on the value of gradients of the form
\begin{equation}
    \frac{\partial C_n(\theta)}{\partial \theta_k}
:=\frac{\partial}{\partial \theta_k}\bra{\psi} G_a(\theta)^{\dagger} H_{n-1}(R) G_a(\theta) \ket{\psi}.
\end{equation}
If the value of these gradients is large for all $R$ around $R_0$, then the gate $G_a$ can be helpful in preparing ground states of the Taylor approximation $H_{n - 1}(R)$ in this region. From Theorem~\ref{thm:a}, this implies that it can aid in calculating accurate energy gradients. Rather than performing this calculation for a large collection of points $R$, we instead employ Eq.~\eqref{eqn:taylor} to express $H_{n - 1}(R)$ as a linear combination of Hamiltonian derivatives. For example, consider the case $n=2$, where we have 

\begin{align}
    \frac{\partial C_2(\theta)}{\partial \theta_k}& = \frac{\partial}{\partial \theta_k}\bra{\psi} G_a(\theta)^{\dagger} H(R_0) G_a(\theta) \ket{\psi} \nonumber\\
    & + \sum_{j = 1}^N (R_j - R_{0j}) \frac{\partial}{\partial \theta_k}\bra{\psi} G_a(\theta)^{\dagger} \frac{\partial H(R_0)}{\partial R_j} G_a(\theta) \ket{\psi}.
\end{align}
There are contributions due to both $H(R_0)$ and $\frac{\partial H(R_0)}{\partial R_j}$, meaning we should select gates depending on their gradient with respect to both terms. This selection procedure will capture all of the gates which have non-negligible gradients relative to $H_2(R)$. The same holds true for higher-order $H_n(R)$, where we compute expectation values with respect to higher-order derivatives of $H(R)$ at $R_0$ as well. 

A naive approach based on the gate-selection strategy outlined above is to evaluate gradients with respect to all of the necessary Hamiltonian derivatives at each step of the adaptive procedure, and then optimize the resulting circuit. This can be unnecessarily costly if there are many additional gates arising from the Hamiltonian derivatives that would otherwise be absent when considering only $H(R_0)$, as in standard adaptive schemes.

An improvement comes from realizing that only the gates that are needed to prepare approximate ground states need to be optimized. From Theorem~\ref{thm:a}, it suffices to add gates such that the circuit is \emph{capable} of preparing ground states of $H_{n - 1}(R)$, but there is no need to optimize them to do so. These gates can simply be added at the end of the circuit with their initial parameters set to zero. They play a role only when computing the expressions for energy gradients, which may require shifting their parameters away from zero, as we will have to compute derivatives with respect to circuit parameters (see Eq.~\eqref{eq:mit1} and Eq.~\eqref{eq:mit2}). Below we describe the tailgating procedure in full:

\begin{enumerate}
    \item Given some Hamiltonian $H(R)$, identify a variational circuit $U(\theta)$ which can approximate the ground state of $H(R)$ at $R_0$. Let $U(\theta^{*}(R))$ denote the optimized circuit.
    \item Let $\mathcal{H}$ be the collection of Hamiltonian derivatives of $H(R)$ up to order $n-1$, evaluated at $R_0$. Similarly, let $\{G_a\}$ be a pool of parameterized quantum gates and for each $G_a$, define $V_a(\theta) = G_a(\theta) U(\theta^{*}(R_0))$. Then for each Hamiltonian derivative $\frac{\partial^m H(R_0)}{\partial j_1 \cdots \partial j_m} \in \mathcal{H}$, compute the magnitudes of the derivatives 
    \begin{align}
    & & \frac{\partial}{\partial \theta_k} \Big\langle \psi(R_0) \Big| V_a(\theta)^{\dagger} \frac{\partial^m H(R_0)}{\partial j_1 \cdots \partial j_m} V_a(\theta) \Big| \psi(R_0)\Big\rangle,
    \end{align}
    evaluated at $\theta = 0$. If at least one is larger than some $\epsilon > 0$, add $G_a$ to the list $L$.
    \item Let $Q(\theta) = \prod_{G_a \in L} G_a(\theta_a)$. Define the new variational circuit $U_{\text{tailgated}}$ to be $U_{\text{tailgated}}(\theta, \theta') = Q(\theta) U(\theta')$, and new optimal parameters to be $(0, \theta^{*}(R))$. 
\end{enumerate}
Assuming $Q(0) = \mathbbm{1}$, we have
    \begin{equation}
    U_{\text{tailgated}}(0, \theta^{*}(R))|0\rangle = U(\theta^{*}(R))|0\rangle = |\psi(R)\rangle,
    \end{equation}
so the new circuit $U_{\text{tailgated}}$ approximates the ground state just as the original circuit did. However, using the previous justification, this circuit is designed to also be capable of preparing ground states of $H_{n - 1}(R)$ in a neighbourhood around $R_0$, and hence also accurate $n$-th order energy derivatives.

We refer to this procedure as \textit{tailgating}, as we are effectively adding a sequence of gates to the end of a circuit $U$, to get a new circuit $U_{\text{tailgated}}$, but are not optimizing any of their corresponding variational parameters. The entire tailgating procedure is summarized in Figure~\hyperref[fig:tail]{2}.

\section{Numerical examples}
\label{sec:numerics}

\begin{table*}[]
    \centering
    \setlength{\tabcolsep}{0.5em} % for the horizontal padding
    \renewcommand{\arraystretch}{2}% for the vertical padding
    \begin{tabular}{ c|c|c|c } 
     \hline
      & Original Circuit & Tailgated Circuit & GAMESS (FCI) \\
      \hline
      & \multicolumn{3}{c}{BeH Frequencies} \\
      \hline
      $\omega_1$ & $5088.77 \ \text{cm}^{-1}$ & $2568.98 \ \text{cm}^{-1}$ & $2569.52 \ \text{cm}^{-1}$ \\
      \hline
      $\omega_2$ & $2301.41 \ \text{cm}^{-1}$ & $2300.68 \ \text{cm}^{-1}$ & $2298.31 \ \text{cm}^{-1}$ \\
      \hline
      $\omega_3$ & $1018.53 \ \text{cm}^{-1}$ & $784.72 \ \text{cm}^{-1}$ & $780.1 \ \text{cm}^{-1}$ \\
      \hline
      $\omega_4$ & $1018.53 \ \text{cm}^{-1}$ & $784.71 \ \text{cm}^{-1}$ & $780.1 \ \text{cm}^{-1}$ \\
      \hline
      & \multicolumn{3}{c}{H$_2$O Frequencies} \\
      \hline
      $\omega_1$ & $9901.80 \ \text{cm}^{-1}$ & $3845.51 \ \text{cm}^{-1}$  & $3812.60 \ \text{cm}^{-1}$ \\
      \hline
      $\omega_2$ & $3606.33 \ \text{cm}^{-1}$ & $3605.40 \ \text{cm}^{-1}$ & $3569.82 \ \text{cm}^{-1}$ \\
      \hline
      $\omega_3$ & $2047.07 \ \text{cm}^{-1}$ & $2043.55 \ \text{cm}^{-1}$ & $2036.99 \ \text{cm}^{-1}$ \\
     \hline\hline
\end{tabular}
    \caption{Calculated normal mode frequencies of BeH$_2$ and H$_2$O using the minimal STO-3G basis set. The parameter $\omega_j$ denotes the $j$-th normal mode frequency. Note that in the quantum methods,
    the BeH$_2$ ground state was prepared with fidelity $> 0.9998$ and the H$_2$O ground state was prepared with fidelity $> 0.9999$, but only the tailgated circuit gives accurate values in both cases.}
    \label{tab:mylabel}
\end{table*}

To highlight the effectiveness of the tailgating procedure, we study a practical example: computing the normal mode vibrational frequencies of molecules. Specifically, we consider BeH$_2$ and H$_2$O at their equilibrium coordinates. Each molecular Hamiltonian is of the form

\begin{equation}
    H(R) = \displaystyle\sum_{pq} h_{pq}(R) c_p^{\dagger} c_q + \frac{1}{2} \displaystyle\sum_{pqrs} h_{pqrs}(R) c_p^{\dagger} c_{q}^{\dagger} c_r c_s,
    \label{eqn:ham}
\end{equation}
where $R$ is a set of nuclear coordinates, $c^{\dagger}_p$ and $c_p$ are the fermionic creation and annihilation operators, acting on the $p$-th orbital, and

\begin{equation}
    h_{pq}(R) = \displaystyle\int dr \ \phi_p^{*}(r) \left( -\frac{\nabla^2}{2} - \displaystyle\sum_{I} \frac{Z_I}{|r - R_I|} \right) \phi_q(r),
    \label{eqn:one_elec}
\end{equation}

\begin{equation}
    h_{pqrs}(R) = \displaystyle\int dr_1 dr_2 \ \frac{\phi_p^{*}(r_1) \phi_q^{*}(r_2) \phi_r(r_2) \phi_s(r_1)}{|r_1 - r_2|},
    \label{eqn:two_elec}
\end{equation}

\noindent
are the one and two-electron integrals yielded from a set of molecular orbitals $\phi_p(r)$, usually obtained by using the Hartree-Fock method~\cite{arrazola2021differentiable}. Note that these molecular orbitals depend implicitly on $R$. The normal mode frequencies of a molecule are precisely the eigenvalues of the energy Hessian -- the matrix of second-order energy derivatives with respect to each of the nuclear coordinates, $\frac{\partial^2 E(R)}{\partial R_{i} \partial R_{j}}$.
\newline

We use the PennyLane library for quantum differentiable programming \cite{bergholm2018pennylane, arrazola2021differentiable} to simulate an adaptive circuit-building procedure, followed by VQE, to compute approximate ground-state energies at the equilibrium geometry. More specifically, we run an adaptive gate selection procedure using a pool of gates composed of all admissible single and double excitation gates for each particular molecule~\cite{arrazola2021universal}, and optimize the resulting circuit by minimizing the energy expectation value of the Hamiltonian, at geometry $R_0$, with gradient descent. The coordinates $R_0$ are the equilibrium geometries of the molecules obtained at the level of Hartree-Fock with the computational chemistry package GAMESS~\cite{GAMESS}. All calculations are performed using the minimal STO-3G basis set.

This process leads to short-depth circuits with optimized parameters that prepare approximate ground states. We then employ the tailgating procedure to adaptively add new gates to the circuit to make it suitable for second-order energy derivatives, which are calculated with the analytic formula derived in Ref.~\cite{mitarai2020theory},

\begin{multline}
\label{eq:mit1}
    \frac{\partial^2 \tilde{E}(R)}{\partial R_i \partial R_j} = \displaystyle\sum_{a} \frac{\partial \theta_a^{*}(R)}{\partial R_i} \frac{\partial}{\partial \theta_a} \Big\langle \psi(\theta) \Big| \frac{\partial H(R)}{\partial R_j} \Big| \psi(\theta) \Big\rangle \biggr\rvert_{\theta = \theta^{*}(R)} \\ + \Big\langle \psi(R) \Big| \frac{\partial^2 H(R)}{\partial R_i \partial R_j} \Big| \psi(R) \Big\rangle.
\end{multline}
The derivatives of the optimal circuit parameters $\theta^{*}(R)$ with respect to the nuclear coordinates are given by the response equation \cite{mitarai2020theory}:

\begin{equation}
\label{eq:mit2}
    \displaystyle\sum_{a} \frac{\partial \theta^{*}_a(R)}{\partial R_i} \frac{\partial^2 \tilde{E}(\theta, R)}{\partial \theta_b \partial \theta_a} \biggr\rvert_{\theta = \theta^{*}(R)} = -\frac{\partial}{\partial \theta_b} \frac{\partial \tilde{E}(R)}{\partial R_i},
\end{equation}
where, by the Feynman-Hellman theorem, the first-order energy derivative is given by

\begin{equation}
    \frac{\partial \tilde{E}(R)}{\partial R_i} = \left\langle \psi(R) \left| \frac{\partial H(R)}{\partial R_i} \right| \psi(R) \right\rangle.
\end{equation}

From here, we calculate the normal mode frequencies for the molecules, comparing the values yielded from utilizing the original variational circuit to a tailgated version of the circuit. To find the accurate values of the normal mode frequencies for reference, we use GAMESS to perform full configuration-interaction (FCI). The results of the calculations are summarized in Table~\hyperref[tab:mylabel]{I}. The frequencies yielded from the tailgated circuit are in much closer agreement with the value given by GAMESS FCI, for both BeH$_2$ and H$_2$O. The full details of the numerics can be found at \url{https://github.com/XanaduAI/tailgating}.
\newline

In addition to the normal mode frequency calculations, we can also return to plotting fidelity curves (as in Fig.~\hyperref[fig:gr]{1}) to gain a more intuitive understanding of the effects of tailgating circuits. In 
Fig.~\hyperref[fig:gr2]{3}, we provide a comparison of the curves of fidelity $| \langle \psi(R_0 + \delta e_{H}') | \psi_0(R_0 + \delta e_{H}') \rangle |^2$ for both the cases where $|\psi(R)\rangle$ is prepared by a non-tailgated and a tailgated circuit. In this example, we consider the molecular Hamiltonian $H(R)$ corresponding to H$_3^{+}$, where $\delta e_{H}'$ is a displacement vector corresponding to vertical stretching (in opposite directions) of two of the hydrogen atoms from their equilibrium configurations at $R_0$. As in Fig.~\hyperref[fig:gr]{1}, we plot the fidelity as a function of $\delta$.

\begin{figure}
    \centering
    \includegraphics[width=250pt]{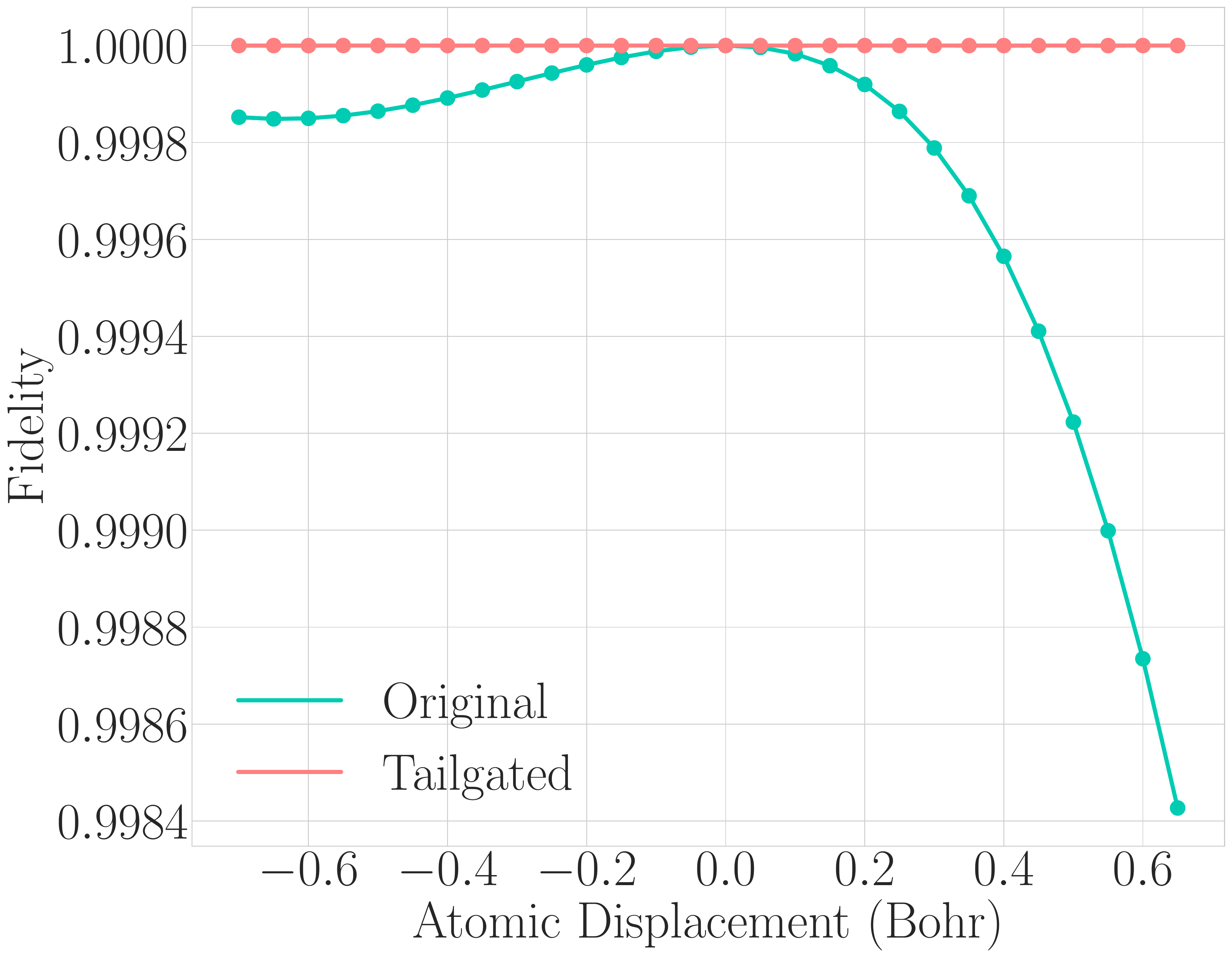}
    \caption{A plot showing the fidelity $|\langle \psi(R_0 + \delta e_H') | \psi_0(R_0 + \delta e_H') \rangle|^2$ between the true ground state of H$_3^{+}$, $|\psi_0(R)\rangle$, and $|\psi(R)\rangle$ prepared 
    by both non-tailgated and tailgated circuits, at $R = \delta e_H'$ for a range of $\delta$, which specifies the displacement of the hydrogens from equilibrium.}
    \label{fig:gr2}
\end{figure}

As can be seen in the figure, the tailgating procedure effectively ``flattens" the curve of fidelities between the prepared state and the true state, which suggests that derivatives of the difference $|\psi(R_0 + \delta e_{H}')\rangle - |\psi_0(R_0 + \delta e_{H}')\rangle$ will be small. This is exactly the property we desire to compute accurate energy derivatives, as was explained in Sec.~\hyperref[sec:background]{II}.
\vspace{-5pt}

\section{Conclusion}

This work outlines possible shortcomings of adaptively-prepared variational quantum circuits when computing higher-order energy derivatives, and proposes a solution which relies on the procedure of tailgating. We first discussed the
problem of computing accurate energy derivatives beyond the first-order with variational quantum circuits, and provide intuition as to why this problem arises. Specifically, we argued that errors in derivatives appear because the relative degree of accuracy to which some variational circuits can prepare correct ground states is large. 
%If we consider the ground state yielded from using a circuit $U$ to minimize the energy expectation value of $H(R)$, as a function $g : R \mapsto |\psi(R)\rangle$, it is not enough to have $g(R_0) = |\psi(R_0)\rangle \approx |\psi_0(R_0)\rangle$: we need $g$ to match the behaviour of $f : R \mapsto |\psi_0(R)\rangle$, near $R_0$. The extent to which these two functions need to have similar behaviour is determined by the order to which we want to accurately compute energy derivatives.

Using this intuition, we then propose sufficient conditions for a variational circuit to yield correct $n$-th order energy derivatives in Theorem~\hyperref[thm:a]{1}.
%, based on a Taylor expansion of the Hamiltonian $H(R)$ around $R_0$. 
The result we prove about $n$-th order energy derivatives and $(n - 1)$-th order Taylor expansions of $H(R)$ may find more general use in other algorithms related to high-order energy derivatives. Using Theorem~\hyperref[thm:a]{1}, we justify the tailgating procedure, in which extra gates are added to the end of an already-optimized variational circuit, and are used when computing derivatives of the circuit with respect to its parameters. 

Finally,
we provide numerical examples demonstrating that tailgating yields better results in certain cases, compared to standard adaptive procedures. We consider the concrete problem of computing normal mode frequencies of BeH$_2$ and H$_2$O molecules,
which requires knowledge of second-order energy derivatives with respect to the atomic coordinates of each molecule. Comparing tailgated and non-tailgated circuits to the true values obtained from classical methods, tailgating can lead to a substantial increase in the accuracy of the calculations. Normal mode frequencies are critical for understanding the vibrational and thermodynamic properties of molecules, as they determine vibronic spectra, vibrational partition functions, and other observable quantities. To use variational circuits for chemistry tasks that involve computing high-order energy derivatives, we argue that a procedure such as tailgating is valuable to ensure the accuracy of the calculations while maintaining the lower cost arising from adaptive procedures. 

We note that our work does not address other important challenges in quantum algorithms for quantum chemistry, which still face many obstacles before they can become competitive with existing classical methods. In particular, estimating high-order energy derivatives using variational algorithms suffers from a more pronounced version of the ``measurement problem", which refers to the often prohibitive number of circuit executions needed to estimate accurate expectation values. Thus, while we tackle issues regarding the quality and gate count of variational circuits, other obstacles remain to be overcome. Ultimately, we hope that tailgating will emerge as a useful member of the toolbox of techniques that scientists can use to construct variational circuits for particular classes of problems in chemistry, and beyond.

%Future directions in which this work can be extended may involve simulating tailgating for larger molecules, and understanding how the number of gates added during tailgating scales with system size. In the case that scaling is poor, finding other techniques for guaranteeing that quantum circuits can compute accurate high-order energy derivatives, with fewer gates and gradient evaluations than tailgating, could be a worthwhile problem to tackle. 

\section{Acknowledgements}

We thank Nathan Wiebe and John Sipe for valuable conversations relating to this project. This research was partially funded by a MITACS Accelerate Grant.

\bibliographystyle{apsrev}
\bibliography{refs.bib}

\appendix
\label{sec:app}

\section{Proof of Theorem 1}

In this section, we present a proof of Theorem~\ref{thm:a}, which was introduced in the main text. We begin by proving a lemma. Note in the proof that follows, for notational convenience, we define

\begin{equation}
    \Bigg| \frac{\partial^{n} \phi(R_0) }{\partial R_{j_1} \cdots \partial R_{j_n}} \Bigg\rangle := \frac{\partial^n | \phi(R_0) \rangle}{\partial R_{j_1} \cdots \partial R_{j_n}}.
\end{equation}

\begin{lem}
\label{lem:1}
Let $H(R)$ be a parameterized Hamiltonian, $R_0$ a set of parameters, and $|\phi(R)\rangle$ a parameterized state vector. Define

\begin{equation}
C(R) = \langle \phi(R) | H(R) | \phi(R) \rangle.
\end{equation}
If $|\phi(R_0)\rangle$ is an eigenstate of $H(R_0)$, then any $n$-th order derivative of $C$ at $R_0$ can be written as a sum of inner products involving $0$-th to $(n - 1)$-th order derivatives of $|\phi(R)\rangle$ evaluated at $R_0$.
\end{lem}

\begin{proof}
It can be checked that the only $n$-th order derivative of $|\phi(R)\rangle$ in an $n$-th order derivative of $C$ appears in a term of the form

\begin{equation}
   2 \text{Re} \left\langle \ \frac{\partial^{n} \phi(R_0)}{\partial R_{j_1} \cdots \partial R_{j_n}} \Bigg| H(R_0) \Bigg| \phi(R_0) \right\rangle.
   \label{eqn:exp}
\end{equation}
Hence, the problem of writing an $n$-th order derivative in terms of $0$-th to $(n - 1)$-th order derivatives of $|\phi(R)\rangle$ is reduced to writing the above expression in terms of such derivatives. We know that

\begin{equation}
    H(R_0) |\phi(R_0)\rangle = C(R_0) | \phi(R_0) \rangle,
\end{equation}

as $|\phi(R_0)\rangle$ is an eigenvector. Hence,

\begin{align}
   &2 \text{Re} \left\langle \ \frac{\partial^{n} \phi(R_0)}{\partial R_{j_1} \cdots \partial R_{j_n}} \Bigg| H(R_0) \Bigg| \phi(R_0) \right\rangle =\nonumber\\
   &2 \text{Re} \ C(R_0) \left\langle \ \frac{\partial^{n} \phi(R_0)}{\partial R_{j_1} \cdots \partial R_{j_n}} \Bigg| \psi(R_0) \right\rangle.
   \label{eqn:r1}
\end{align}

Finally, since $|\phi(R)\rangle$ is normalized, we have $\langle \phi(R) | \phi(R) \rangle = 1$ for all $R$, which implies that any $k$-th order derivative, for $k \geq 1$, of this inner product will be equal to $0$. It follows that

\begin{align}
    &\frac{\partial^{n} \langle \phi(R_0) | \phi(R_0) \rangle}{\partial R_{j_1} \cdots \partial R_{j_n}}\nonumber \\
    &= 2 \text{Re} \left\langle \ \frac{\partial^{n} \phi(R_0)}{\partial R_{j_1} \cdots \partial R_{j_n}} \Bigg| \phi(R_0) \right\rangle + S_{n - 1}(\phi(R_0)) = 0,
    \label{eqn:r2}
\end{align}

\noindent
where $S_{n - 1}(\phi(R_0))$ is a sum of inner products involving $0$-th to $(n - 1)$-th order derivatives of $|\phi(R_0)\rangle$. Rearranging, and using Eq.~\eqref{eqn:r1} and Eq.~\eqref{eqn:r2}, gives that the expression in Eq.~\eqref{eqn:exp} is equal to $-2 \text{Re} [C(R_0) S_{n - 1}(\phi(R_0))]$: an expression involving only $0$-th to $(n - 1)$-th order derivatives of $|\phi(R_0)\rangle$. This completes the proof.
\end{proof}

\begin{cor}
If for a circuit $U$, the corresponding $p$-th order derivatives of $|\psi(R)\rangle$ match those of $|\psi_0(R)\rangle$ for $ 0 \leq p \leq n - 1$ at $R_0$, then $U$ can yield accurate $n$-th order energy derivatives at $R_0$.
\label{cor:a}
\end{cor}

\begin{proof}
Recall the definitions of $|\psi(R)\rangle$ and $\tilde{E}(R)$ from Sec.~\hyperref[sec:background]{II}. Since, when $n \geq 1$, both $|\psi(R)\rangle$ and $|\psi_0(R)\rangle$ satisfy the conditions of Lemma~\hyperref[lem:1]{1}, it follows that both $\frac{\partial^{p} \tilde{E}(R_0)}{\partial R_{j_1} \cdots \partial R_{j_p}}$ and $\frac{\partial^{p} E(R_0)}{\partial R_{j_1} \cdots \partial R_{j_p}}$ can be written in terms of inner products of $0$-th to $(p - 1)$-th order derivatives of $|\psi(R)\rangle$ and $|\psi_0(R)\rangle$ respectively. By assumption, we have 

\begin{equation}
\frac{\partial^{p} | \psi(R_0) \rangle}{\partial R_{j_1} \cdots \partial R_{j_p}} = \frac{\partial^{p} | \psi_0(R_0) \rangle}{\partial R_{j_1} \cdots \partial R_{j_p}}
\end{equation}

\noindent
for all $p$ from $1$ to $n - 1$, and all choice of coordinates, $\{R_{j_1}, \ \ldots, \ R_{j_p}\}$. Thus, the representations of the $n$-th order derivatives of $\tilde{E}$ and $E$ in terms of such derivatives must be equal as well.
\end{proof}

\begin{lem}
 Let $|\phi(R)\rangle$ be the ground state of $H_{n - 1}(R)$. The derivatives of $|\psi_0(R)\rangle$ and $|\phi(R)\rangle$ agree to order $n - 1$ when evaluated at $R_0$.
\end{lem}

\begin{proof}
Note that
\vspace{-5pt}

\begin{multline}
    H(R) = H_{n - 1}(R) \\ + \frac{1}{n!} \displaystyle\sum_{j_1, \ldots, j_n = 1}^N \frac{\partial^n H(R_0)}{\partial R_{j_1} \cdots \partial R_{j_n}} (R_{j_1} - R_{0j_1}) \cdots (R_{j_n} - R_{0j_n}) \\ + \cdots
    \label{eqn:d}
\end{multline}
If $H(R)$ is a non-degenerate Hamiltonian with normalized eigenvectors $|v_j(R)\rangle$ and corresponding eigenvalues $\lambda_j(R)$, then

\begin{equation}
\frac{\partial |v_j(R)\rangle}{\partial R_k} = \displaystyle\sum_{i \neq j} \frac{\left\langle v_j(R) \left| \frac{\partial H(R)}{\partial R_k} \right| v_i(R) \right\rangle}{\lambda_j(R) - \lambda_i(R)} |v_i(R)\rangle
\label{eqn:per}
\end{equation}
(see Appendix~\ref{sec:appb}), and
\begin{equation}
    \frac{\partial \lambda_j(R)}{\partial R_k} = \left\langle v_j(R) \left| \frac{\partial H(R)}{\partial R_k} \right| v_j(R) \right\rangle,
    \label{eqn:per2}
\end{equation}
from the Feynman-Hellmann theorem. It follows from Eq.~\eqref{eqn:per} that $\frac{\partial^{n} |v_j(R)\rangle}{\partial R_{k_1} \cdots \partial R_{k_n}}$ can be written as a sum of terms involving $(n - 1)$-th order derivatives of the eigenvectors and eigenvalues, and $n$-th order derivatives of the Hamiltonian. In addition, it follows from Eq.~\eqref{eqn:per2} that $(n - 1)$-th order derivatives of the eigenvalues can be written in terms of $(n - 1)$-th order derivatives of the Hamiltonian and $(n - 2)$-th order derivatives of the eigenvectors.
\newline

Thus, continuing inductively, it follows that $\frac{\partial^{n - 1} |v_j(R)\rangle}{\partial R_{k_1} \cdots \partial R_{k_{n - 1}}}$ can be written in terms of the eigenvectors $|v_i(R)\rangle$, eigenvalues $\lambda_i(R)$, and $0$-th to $(n - 1)$-th order derivatives of the Hamiltonian.
\newline

Since $H(R_0) = H_{n - 1}(R_0)$, so their eigenvectors and eigenvalues at $R_0$ are the same, and from Eq.~\eqref{eqn:d} the derivatives of the two Hamiltonians agree to the $(n - 1)$-th order at $R_0$, it follows that the derivatives of $|\psi_0(R)\rangle$ and $|\phi(R)\rangle$ agree to order $n - 1$ at $R_0$.
\end{proof}

Now, we are able to prove the desired result.
\newline

\noindent
\textbf{Theorem 1.} \textit{Let $|\phi(R)\rangle$ be the ground state of $H_{n - 1}(R)$. Suppose that there exists $\delta > 0$ such that a circuit $U(\theta)$ can prepare the state $|\phi(R)\rangle$ for all $R$ such that $0 < \| R - R_0 \| < \delta$. Then it holds that}

\begin{equation}
\frac{\partial^{p} E(R_0)}{\partial R_{j_1} \cdots \partial R_{j_p}} = \frac{\partial^{p} \tilde{E}(R_0)}{\partial R_{j_1} \cdots \partial R_{j_p}}
\end{equation}

\noindent
\textit{for all $p$ such that $0 \leq p \leq n$, and all choices of coordinates, $\{R_{j_1}, \ ..., \ R_{j_p}\}$.}

\begin{proof}
Recall that $|\psi(R)\rangle$ is the state such that $U(\theta(R))$ closest approximates $|\psi_0(R)\rangle$ (with respect to the vector $2$-norm). Thus, we have that

\begin{equation}
\| |\psi(R)\rangle - |\psi_0(R)\rangle \| \leq \| |\phi(R)\rangle - | \psi_0(R)\rangle \|
\end{equation}
when $0 \leq ||R - R_0|| < \delta$, as $U$ can also prepare $|\phi(R)\rangle$. It follows immediately that the derivatives of $|\psi(R)\rangle$ and $|\psi_0(R)\rangle$ will agree to order $n - 1$ at $R_0$, as those of $|\phi(R)\rangle$ and $|\psi_0(R)\rangle$ agree. We then apply Corollary~\hyperref[cor:a]{1} to arrive at the result.
\end{proof}

\section{Derivatives of Eigenvectors}
\label{sec:appb}

In this section, we briefly outline how to compute derivatives of eigenvectors corresponding to a non-degenerate Hamiltonian. Suppose $H(R)$ is a Hamiltonian, such that

\begin{equation}
    H(R) |v_i(R)\rangle = \lambda_i(R) |v_i(R)\rangle,
    \label{eq:eig}
\end{equation}

where we assume $\lambda_i(R)$ and $|v_i(R)\rangle$ are smooth functions such that the set of eigenvectors $|v_i(R)\rangle$ forms an orthonormal basis, which is guaranteed by the spectral theorem. Taking the derivative of both sides of Eq.~\eqref{eq:eig} an re-arranging, we get
\vspace{-5pt}

\begin{equation}
    \left[ \lambda_i(R) - H(R) \right] \frac{\partial |v_i(R)\rangle}{\partial R_k} = \left[ \frac{\partial H(R)}{\partial R_k} - \frac{\partial \lambda_i(R)}{\partial R_k} \right] |v_i(R)\rangle.
\end{equation}

Taking the inner product of both sides of the above equation with $|v_j(R)\rangle$, it follows that

\begin{multline}
    (\lambda_i(R) - \lambda_j(R)) \Big\langle v_j(R) \Big| \frac{\partial v_i(R)}{\partial R_k} \Big\rangle = \\ \Big\langle v_j(R) \Big| \frac{\partial H(R)}{\partial R_k} \Big| v_i(R) \Big\rangle - \frac{\partial \lambda_i(R)}{\partial R_k} \langle v_j(R) | v_i(R) \rangle.
\end{multline}

In the case that $i \neq j$, then $\langle v_j(R) | v_i(R) \rangle = 0$, and we can re-arrange to get

\begin{equation}
    \Big\langle v_j(R) \Big| \frac{\partial v_i(R)}{\partial R_k} \Big\rangle = \frac{\Big\langle v_j(R) \Big| \frac{\partial H(R)}{\partial R_k} \Big| v_i(R) \Big\rangle}{\lambda_i(R) - \lambda_j(R)}
\end{equation}

In the case that $i = j$, then since $\langle v_i(R) | v_i(R) \rangle = 1$ for all $R$, we will get

\begin{equation}
    \frac{\partial}{\partial R_k} \langle v_i(R) | v_i(R) \rangle = 2 \text{Re} \left[ \Big\langle v_i(R) \Big| \frac{\partial v_i(R)}{\partial R_k} \Big \rangle \right] = 0
\end{equation}

In general, we can assume that $\Big\langle v_i(R) \Big| \frac{\partial v_i(R)}{\partial R_k} \Big\rangle$ is real, as we can always multiply either state by a non-physical global phase. 
Hence, this inner product will be equal to $0$. Therefore, since the $|v_j(R)\rangle$ form a basis, we will have

\begin{equation}
\frac{\partial | v_i(R) \rangle}{\partial R_k} = \displaystyle\sum_{j \neq i} \frac{\Big\langle v_j(R) \Big| \frac{\partial H(R)}{\partial R_k} \Big| v_i(R) \Big\rangle}{\lambda_i(R) - \lambda_j(R)} |v_j(R)\rangle.
\end{equation}

\end{document}